\documentclass[AMA,STIX1COL]{WileyNJD-v2}

\articletype{Research Article}%

\received{26 April 2016}
\revised{6 June 2016}
\accepted{6 June 2016}

\usepackage{mathtools}

\usepackage{algorithm}
\usepackage{algpseudocode}
\usepackage{comment}
\usepackage{amsmath}
\usepackage[firstpageonly=true]{draftwatermark}

\newcommand{\rn}[1]{%
  \textup{\lowercase\expandafter{\romannumeral#1}}%
}
\newcommand{\REV}[1]{#1}
\newcommand{\REVFB}[1]{#1}
\begin{document}

\title{Robust offset-free nonlinear model predictive control for systems \\ learned by neural nonlinear autoregressive exogenous models}

\author[1]{Jing Xie}

\author[1]{Fabio Bonassi}

\author[1]{Marcello Farina}

\author[1]{Riccardo Scattolini}

\authormark{Xie \textsc{et al}}


\address{\orgdiv{Dipartimento di Elettronica, \\ Informazione e Bioingegneria}, \\ \orgname{Politecnico di Milano}, \orgaddress{\state{Milano}, \country{Italy}}}

\corres{Jing Xie, Dipartimento di Elettronica, Informazione e Bioingegneria, Politecnico di Milano, Via Ponzio 34/5, 20133, Milano, Italy. \\ \email{jing.xie@polimi.it}}

\fundingInfo{European Union’s Horizon 2020 research and innovation programme under the Marie Skłodowska-Curie grant agreement \\ No. 953348}

\abstract[Abstract]{This paper presents a robust Model Predictive Control (MPC) scheme that provides offset-free setpoint tracking for systems described by Neural Nonlinear AutoRegressive eXogenous (NNARX) models.
To this end, a NNARX model that learns the dynamics of the plant from input-output data is augmented with an explicit integral action on the output tracking error.
A robust tube-based MPC is finally designed, leveraging the unique structure of the model, to ensure \REV{robust offset-free convergence to constant reference signals even in case of plant-model mismatch}.
Numerical simulations on a water heating system show the effectiveness of the proposed control algorithm.}

\keywords{nonlinear model predictive control, offset-free tracking, neural networks, robust control, learning-based control}
\maketitle

 \DraftwatermarkOptions{%
 angle=0,
 hpos=0.5\paperwidth,
 vpos=0.95\paperheight,
 fontsize=0.012\paperwidth,
 color={[gray]{0.2}},
 text={
   \parbox{0.99\textwidth}{This is the peer reviewed version of \href{http://doi.org/10.1002/rnc.6883}{DOI: 10.1002/rnc.6883}, available in Open Access on the publisher's website. Please, cite Reference \citenum{xie2023robust} instead. \\
   This article may be used for non-commercial purposes in accordance with Wiley Terms and Conditions for Use of Self-Archived Versions. This article may not be enhanced, enriched or otherwise transformed into a derivative work, without express permission from Wiley or by statutory rights under applicable legislation. Copyright notices must not be removed, obscured or modified. The article must be linked to Wiley’s version of record on Wiley Online Library and any embedding, framing or otherwise making available the article or pages thereof by third parties from platforms, services and websites other than Wiley Online Library must be prohibited.}},
 }

\section{Introduction}\label{sec1}

Learning algorithms for model identification and control design are increasingly popular in the control community  \cite{zhong2013learningbased,schoukens2019nonlinear}.
Among the main approaches developed so far, we here recall Gaussian processes\cite{hewing2019cautious}, set membership identification algorithms\cite{terzi2019learning}, and Bayesian identification\cite{GREEN2015109}, where large and informative data-sets are utilized to extract important information on the characteristics of the system under control.
 In this framework, Recurrent Neural Networks (RNN) are also widely used, see  References \citenum{dynamicrnn1995delgado} and \citenum{nncontrol1995}, due to their ability to describe nonlinear dynamics. 
With the goal of learning safe RNN models, significant research efforts have been recently devoted to the analysis of the stability properties of several RNN structures such as Gated Recurrent Units (GRU)\cite{bonassi2020stability},
Short Term Memory networks (LSTM)\cite{terzi2021learning}, Echo State Networks (ESN)\cite{armenio2019model}, and Neural Nonlinear AutoRegressive eXogenous models (NNARX)\cite{bonassi2021nnarx}. 
In these works, sufficient conditions for Input-to-State Stability (ISS) and Incremental Input-to-State Stability ($\delta$ISS) of the adopted RNN model have been derived. 
These conditions have been exploited for training provenly stable RNN models, which, in turn, has been shown to pave the way to theoretically-sound control strategies, mainly relying on Model Predictive Control \cite{terzi2021learning, armenio2019model}, albeit other control strategies have also been proposed \cite{bonassi2022recurrent}. 
An overview of the advantages of RNN with stability properties is reported in Reference \citenum{bonassi2022survey}, together with an in-depth discussion on the many issues to be considered when using RNN for control design, such as the  verifiability, interpretability, and lifelong learning adaptation.

A further step towards the application of MPC with RNN models concerns the development of control schemes guaranteeing asymptotic zero error regulation for constant reference signals also in the presence of exogenous disturbances. 
To this end, several offset-free tracking MPC algorithms have been proposed both for linear and nonlinear plant models. 
In Reference \citenum{magni2001output}, the model is augmented with an integrator on the output tracking error and an observer is designed to reconstruct the state of the system. 
In this solution, offset-free is achieved without defining and estimating any disturbance. 
This approach is adopted also in Reference \citenum{bonassi2021nonlinear}, where an offset-free MPC control scheme is designed for systems learned by GRU networks. 
A second, very popular approach is proposed in Reference \citenum{morari2012nonlinear}, where the system is augmented with a (fictitious) disturbance model. 
The estimation of this disturbance is then used to compute the steady-state values of the system's states and inputs which guarantee zero error regulation also in case of plant-model mismatch.

Regardless of the problem of asymptotic zero-error regulation, model uncertainties and/or perturbations can cause the violation of state and input constraints.
In addition, nominal stability can be lost, which can pose huge risks for safety-critical systems.
In this context, several robust MPC algorithms have been proposed, such as min-max MPC\cite{minmax1998scokaert}, where the optimal control sequence is computed by minimizing the cost function in the worst case scenario.
However, this method generally results in a computationally intractable optimization problem. 
\REV{
An effective alternative solution is that of the popular tube-based approaches\cite{MAYNE2005robust,LANGSON2004tube} which, albeit being mainly developed for linear systems, have also been recently extended to the nonlinear ones.}


In this paper, we propose a solution to the \emph{robust zero-error regulation} problem for systems described by Neural NARX models.
\REV{Compared to other RNNs, the recurrence of these models only comes from the feedback of past input-output data instead of from hidden neurons, which results in a simpler structure and easier training procedures. 
The state of NNARXs consists only of past input-output data, so that a state observer is not needed in the implementation of the MPC algorithm. For those reasons, Neural NARX models have been widely used in academic research \cite{nn1996levin, nn1993levin} and industrial applications\cite{nnapplication2008himmelblau, MOHDALI2015}.} 
Along the lines of References \citenum{magni2001output} and \citenum{bonassi2021nonlinear}, we propose to augment the control structure with two elements: (\emph{i}) an integrator on the output tracking error, to achieve asymptotic offset-free tracking; (\emph{ii}) a derivative action on the MPC output, useful to apply almost standard stability results for nominal and robust "tube-based" MPC. Some preliminary results, obtained by considering only the nominal case, have been reported in Reference \citenum{offsetnarx2022fabio}.

The proposed approach has been tested on the model of a water heating benchmark system, with two main goals.
First, to assess the closed-loop performances of the nominal MPC law, and to compare them to those achieved by the strategy proposed in reference \citenum{morari2012nonlinear}: results suggest that the proposed approach achieves remarkable performances even in presence of disturbances.
Second, to assess the robustness of the tube-based MPC control strategy against exogenous disturbances.

The paper is structured as follows.
In Section \ref{sec:statement} the NNARX nominal model and the perturbed system are introduced.
In Section \ref{sec:control} the proposed robust control framework is described in detail.
Then, in Section \ref{sec:example} the proposed control architecture is tested on a simulated water heating benchmark system.
Lastly, some conclusions are drawn in Section \ref{sec:conclusions}.

\subsection{Notation}

The following notation is adopted in this paper. Let $R$ denote the field of real numbers, $R^n$ the n-dimensional Euclidean space. Given a vector $v$, $v^\prime$ is its transpose, $\| v \|_p$ its $p$-norm and $\|p\|_{\infty}$ its maximum norm.
In addition, given a matrix $Q$, we denote $\| v \|_Q^2 := v^\prime Q v$.
Sequences of vectors are indicated by bold-face fonts, i.e., $\boldsymbol{v}_k = \left\{ v_0, ..., v_k \right\}$. Given two sets $\mathcal{A}$ and $\mathcal{B}$, $\mathcal{A}\oplus\mathcal{B} := \left\{a+b| a\in\mathcal{A}, b\in \mathcal{B}\right\}$ denotes the Minkowski set addition and $\mathcal{A} \ominus \mathcal{B} := \left\{a| a\oplus \mathcal{B} \subseteq \mathcal{A}\right\}$ the Pontryagin set subtraction.

\section{Problem statement and NNARX Model} \label{sec:statement}
The robust tracking problem we want to solve can be described as follows. Consider a general dynamical system
\begin{eqnarray}\label{sist}
  y=\mathcal{S}(u),
\end{eqnarray}
with input $u \in \mathcal{U}\subset R^m$ and output $y \in \mathcal{Y}\subset R^p$, where $\mathcal{U}$ and $\mathcal{Y}$ are closed sets. For simplicity, in the following it is assumed that the number of outputs $p$ equals the number of inputs $m$, but the proposed approach can be easily extended to the case $m > p$.

Assume that the nominal model of the plant is described by
\begin{eqnarray}\label{mod}
  \bar{y}=\mathcal{\bar{S}}(u),
\end{eqnarray}
where the output $\bar{y} \in \mathcal{Y}\subset R^p$.\REV{We assume that the amplitude of the modeling error is quantified (as specified later).}
The aim of this work is to design for $\mathcal{S}$, based on  $\mathcal{\bar{S}}$, an MPC control law guaranteeing closed-loop stability and robust asymptotic zero error regulation for constant reference signals $y^o \in \mathcal{{Y}}\subset R^p$.

In this paper, the above-mentioned problem is tackled in three steps:
\begin{enumerate}[i.]
  \item definition of the nominal NNARX model $\mathcal{\bar{S}}$ and of the modeling error;
  \item design of a control architecture \REV{aiming to guarantee asymptotic zero-error output regulation in presence of plant-model mismatch};
  \item definition of a robust MPC algorithm guaranteeing constraints' satisfaction and stability results also in the perturbed case $\mathcal{\bar{S}}\neq \mathcal{S}$.
\end{enumerate}

\subsection{Neural NARX nominal model} \label{sec:nnarx}
The nominal model $\mathcal{\bar{S}}$ is assumed to be described by a NNARX architecture\cite{bonassi2021nnarx}, where the output $\bar{y}_{k+1}$ is computed as a regression over past $N$ input and output samples as well as current input $u_k$:
\begin{equation} \label{eq:rnn:nnarx_model}
\bar{y}_{k+1} = \eta(\bar{y}_{k}, \bar{y}_{k-1}, ..., \bar{y}_{k-N+1}, u_{k}, u_{k-1}, ... u_{k-N}; \Phi),
\end{equation}
where $\eta$ denotes the nonlinear regression function and $\Phi$ is the set of model parameters.

Letting the state component be defined as
\begin{equation} \label{eq:rnn:nnarx_states}
	z_{i, k} = \left[\begin{array}{c}
		\bar{y}_{k-N+i} \\
		u_{k-N-1+i}
	\end{array}\right],
\end{equation}
for any $i \in \{1, ..., N \}$, the model \eqref{eq:rnn:nnarx_model} can be given a nonminimal discrete-time state space form
\begin{equation} \label{eq:rnn:nnarx_normal_form}
\begin{dcases}
		z_{1, k+1} = z_{2, k} \\
		\quad \vdots \\
		z_{N-1,k+1} = z_{N, k} \\
		z_{N, k+1} = \begin{bmatrix}
			\eta(z_{1, k}, z_{2, k}, ..., z_{N, k}, u_{k}; \Phi) \\
			u_{k}
		\end{bmatrix} \\
	\bar{y}_{k} = [I\quad 0] \, z_{N, k}
\end{dcases}.
\end{equation}

\begin{subequations} \label{eq:nnarx:statespace}
Let us now define the nominal state vector as
\begin{equation} \label{eq:nnarx:statespace:state}
	\bar{x}_{k} = [ z_{1, k}^\prime, ..., z_{N, k}^\prime]^\prime \in \mathbb{R}^{n},
\end{equation}
where $n = N(m+p)$. 
Note that, in view of the previous definitions, it holds that $\bar{x}\in \mathcal{\bar{X}}$, where $\mathcal{\bar{X}}$ can be easily expressed in terms of $\mathcal{{U}}$ and $\mathcal{{Y}}$.
Model \eqref{eq:rnn:nnarx_normal_form} can be compactly reformulated as
\begin{equation}
\begin{dcases}
  \bar{x}_{k+1} = {A} \bar{x}_{k} + B_{u} u_{k} + B_{x} \eta(\bar{x}_{k}, u_{k}) \\
  \bar{y}_{k} = C \bar{x}_{k}
\end{dcases}
\end{equation}
where $A$, $B_u$, $B_x$, and $C$ are:
\begin{equation}\label{eq:nnarx:statespace:matrices}
	\begin{aligned}
		A = {\begin{bmatrix}
		0 & I & 0 & ... & 0 \\
		0 & 0 & I & ... & 0 \\
		\vdots &&& \ddots & \vdots \\
		0 & 0 & 0 & ... & I \\
		0 & 0 & 0 & ... & 0
		\end{bmatrix}}, \quad
		B_u={\begin{bmatrix}
			0 \\
			0 \\
			\vdots \\
			0 \\
			\tilde{B}_u
		\end{bmatrix}},\quad
		B_x={\begin{bmatrix}
			0 \\
			0 \\
			\vdots \\
			0 \\
			\tilde{B}_x
		\end{bmatrix}},\quad
		C= \begin{bmatrix}
			0 & ... & 0 & \tilde{C}
	\end{bmatrix}.
	\end{aligned}
\end{equation}
$0$ and $I$ are null and identity matrices of proper dimensions, and the sub-matrices $\tilde{B}_u$, $\tilde{B}_{x}$,and $\tilde{C}$ are defined as
\begin{equation}
	\tilde{B}_u = \begin{bmatrix}
		0 \\
		I
	\end{bmatrix}, \quad
	\tilde{B}_x = \begin{bmatrix}
		I \\
		0
	\end{bmatrix}, \quad
	\tilde{C} = \begin{bmatrix}
		I & 0
	\end{bmatrix}.
\end{equation}
\end{subequations}

%
\begin{subequations}
\REV{
The regression function $\eta$ in \eqref{eq:nnarx:statespace} is a Feed-Forward Neural Network (FFNN) consisting of $M$ layers of neurons, where each layer is a linear combination of its inputs, followed by a nonlinear activation function. 
Denoting by $\eta_l$ the output of the $l$-th layer, with $l \in \{ 1, ..., M\}$,
\begin{equation}  \label{eq:model:ffnn2}
	\begin{aligned}
		\eta_1 &= \sigma_1 \big( W_1 u_{k} + U_1 x_k + b_1 \big), \\
		\eta_2 &= \sigma_2 \big( W_2 u_{k} + U_2 \eta_1 + b_2 \big), \\
		&\vdots \\
		\eta_M &= \sigma_M \big( W_M u_{k} + U_M \eta_{M-1} + b_M \big), \\
	\end{aligned}
\end{equation}
where $\sigma_l$ denotes the activation function, assumed to be Lipschitz-continuous and zero-centered, i.e., $\sigma_l(0) = 0$.
The regression function $\eta$, whose structure is illustrated in Figure \ref{fig:nnarx}, thus reads as
\begin{equation}  \label{eq:model:ffnn}
		\eta(x_{k}, u_{k}) = U_0 \eta_M + b_0
\end{equation}}
\end{subequations}

\begin{figure}[bt]
	\centering	\includegraphics[width=\linewidth]{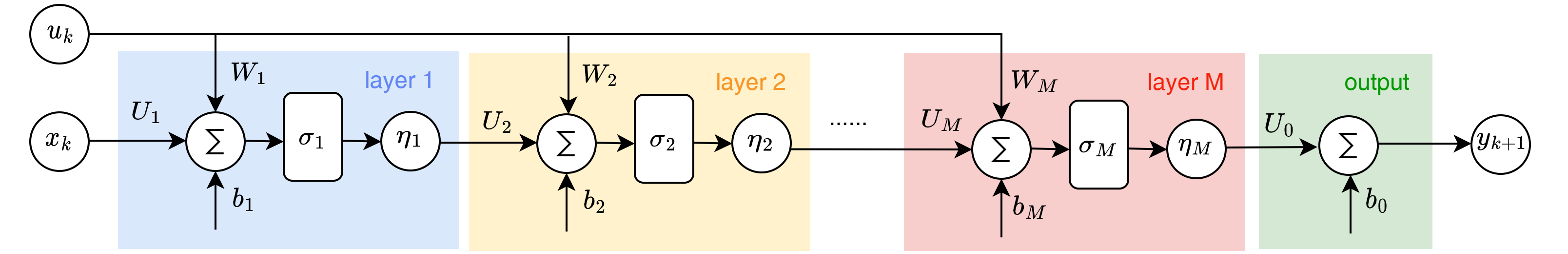}
	\caption{Structure of NNARX model}
	\label{fig:nnarx}
\end{figure}

The matrices $W_l$, $U_l$ and $b_l$ are the weights of each layer, which are training parameters of the network $\Phi = \{ U_0, b_0, \{ U_l, W_l, b_l \}_{ l = 1, ..., M} \}$.
Omitting $\Phi$ for compactness, the NNARX model \eqref{eq:nnarx:statespace} can be written as
\begin{equation} \label{eq:nnarx:compact}
	\mathcal{\bar{S}}: \, \begin{dcases}
		\bar{x}_{k+1} = f(\bar{x}_{k},u_{k}) \\
		\bar{y}_k = C \bar{x}_{k}
	\end{dcases}.
\end{equation}

%

\REV{The goal of the training procedure is to find the set of weights minimizing a performance criterion, known as loss function, which is in general either the prediction or simulation Mean Square Error (MSE).
It is worth pointing out that, as discussed in References \citenum{bonassi2021nnarx} and \citenum{bonassi2022survey}, desirable stability properties, such as the $\delta$ISS property, can also be enforced during the training procedure of the network.
For more details, the interested reader is referred to Appendix \ref{appendix:stability}.}

\subsection{Perturbed Model}
Let the real system dynamics be described by
\begin{equation} \label{eq:nnarx:real}
\mathcal{S}: \begin{dcases}
  {x}_{k+1} = {A} {x}_{k} + B_{u} {u}_{k} + B_x\eta({x}_{k}, {u}_{k})+ B_x\delta({x}_k,{u}_k)\\
  {y}_{k} = C {x}_{k}
\end{dcases}
\end{equation}
where the state $x_k$ enjoys the same structure of the nominal system's state $\bar{x}_k$, \REV{see \eqref{eq:nnarx:statespace:state}, i.e.,
\begin{equation*}
	x_k =  \big[ y_{k-N+1}^{\prime}, u_{k-N}^{\prime}, ..., y_{k}^{\prime}, u_{k-1}^\prime \big]^{\prime}.
\end{equation*}
Note that the perturbed system's dynamics are characterized by an unknown uncertainty $\delta({x}_k,{u}_k)$ that affects the component of $x_k$ that corresponds to $y_k$.}
Let us therefore define the state error $e_k=x_k-\bar{x}_k$ as the deviation between the real state $x_k$ and the nominal one $\bar{x}_k$. 
From \eqref{eq:nnarx:statespace} and \eqref{eq:nnarx:real} it can be obtained that
\begin{equation} \label{eq:errordynamics}
\begin{aligned}
  	e_{k+1} &= Ae_{k} + \underbrace{B_{x} \big[ \eta({x}_{k}, {u}_{k})- \eta(\bar{x}_{k}, u_{k}) + \delta({x}_k,{u}_k) \big] }_{w_k},
\end{aligned}
\end{equation}
which consists of a linear component, $Ae_{k}$, forced by the uncertainty-related term $w_k$, defined as
\begin{equation}\label{eq:perturbation_def}
	w_k = B_{x} \big[ \eta(x_k,{u}_k)- \eta(\bar{x}_k,u_k)+\delta({x}_k,{u}_k) \big].
\end{equation}
\REV{In what follows, it is assumed that such term is bounded $\forall k$ as $w_k \in \mathcal{W}$, where $\mathcal{W}$ is a closed and compact set which depends on $\mathcal{X}$ and $\mathcal{U}$.}
\begin{assumption} \label{distbound}
	The disturbance $w_k$ is bounded in a compact set $\mathcal{W}$ of amplitude $\check{w}$, i.e. for any $k$
\begin{equation} \label{eq:mismatch:upperbound}
	w_k \in \mathcal{W} = \{ w \in \mathbb{R}^n :  \| w \|_{\infty} \leq \check{w} \}
\end{equation}
\end{assumption}

Let us point out that, in the spirit of data-driven control approaches as the one herein described, the value of $\check{w}$ can be estimated from data.
	Indeed, owing to the particular structure of the matrix $B_x$, see \eqref{eq:nnarx:statespace}, the disturbance $w_k$ only affects the entry of $x_{k+1}$ associated to $y_{k+1}$, which makes it easy to estimate $\check{w}$.


\section{Controller design} \label{sec:control}
As previously stated, the goal of this paper is to design a control system such that the controlled output tracks a given constant reference signal ${y^o}$ in a robust way, i.e.
\begin{equation}\label{eq:setp}
    \varepsilon_k = {y^o} - y_k \xrightarrow[k \to \infty]{} 0.
\end{equation}
also in the presence \REV{of the unknown perturbation $\delta(x_k, u_k)$.}

\subsection{The tracking problem in the nominal case} \label{nominal}
To introduce the proposed control structure, we first consider the control design in the nominal case.
To this end, let $({\bar{x}(y^o)},\bar{u}(y^o))$ be a feasible equilibrium of \REVFB{\eqref{eq:nnarx:statespace}}, meaning that $\bar{x}(y^o)\in \mathcal{X}$ and $\bar{u}(y^o)\in \mathcal{{U}}$, whereisfy
\begin{equation}\label{eq:equilibrium}
	\begin{dcases}
		\bar{x}({y^o}) = f(\bar{x}(y^o), \bar{u}(y^o)) \\
  		{y}^o = C \bar{x}(y^o)
	\end{dcases}.
\end{equation}
Then, consider the linearization of the system \eqref{eq:nnarx:compact} around such equilibrium, defined by the matrices
\begin{equation} \label{eq:control:linearization}
\begin{aligned}
	A_\delta = \left. \frac{\partial f(x, u)}{\partial x}  \right\lvert_{(\bar{x}(y^o), \bar{u}(y^o))},  \,\, B_\delta = \left. \frac{\partial f(x, u)}{\partial u}  \right\lvert_{(\bar{x}(y^o), \bar{u}(y^o))}.
\end{aligned}
\end{equation}
\REV{We introduce the following assumption about the linearized model.} 
\REV{\begin{assumption} \label{ass:linearized_gs}
For any equilibrium point $(\bar{x}(y^o), \bar{u}(y^o))$, the linearized system defined by the matrices in \eqref{eq:control:linearization} is asymptotically stable.
\end{assumption}
However, guaranteeing that a generic NNARX model satisfies Assumption \ref{ass:linearized_gs} is not trivial, as it would involve assessing the NNARX local asymptotic stability around an indefinite number of equilibria.
On the other hand, if the NNARX model is trained to be provenly $\delta$ISS (e.g., with the method proposed in Reference \citenum{bonassi2021nnarx}), it is possible to show that Assumption \ref{ass:linearized_gs} is satisfied around any possible equilibrium. This is discussed in the following proposition.}

\begin{proposition} \label{prop:asymptotic_stability}
Consider a nonlinear system in the form \eqref{eq:nnarx:compact}, and assume that it is exponentially $\delta$ISS (Definition \ref{def:exponential_deltaiss}) and that the gradient of $f(x, \bar{u}(y^o))$ with respect to $x$ is Lipschitz continuous.
Then, for each feasible equilibrium $(\bar{x}(y^o), \bar{u}(y^o))$ satisfying \eqref{eq:control:linearization}, matrix $A_\delta$ is Schur stable.
\end{proposition}
\begin{proof}
	See the Appendix.
\end{proof}

In addition, the following Assumption on the linearized system matrices is also introduced.

\begin{assumption} \label{ass:linearized}
	The tuple $(A_\delta, B_\delta, C)$  is reachable, observable, and does not have invariant zeros equal to 1.
\end{assumption}

Under Assumption \ref{ass:linearized} and in light of Theorem 1 in Reference \citenum{de1997stabilizing}, one can guarantee the existence of an open neighborhood of ${y}^o$, denoted by $\Gamma({y}^o) \subseteq \mathcal{R}^{p}$, where, for any $\tilde{y} \in \Gamma({y}^o)$, there exists an equilibrium $(\tilde{x}(\tilde{y}), \tilde{u}(\tilde{y}))$ satisfying the state and input constraints and such that
\begin{equation}
	\begin{dcases}
		\tilde{x}(\tilde{y}) = f(\tilde{x}(\tilde{y}), \tilde{u}(\tilde{y})) \\
  		\tilde{y} = C \tilde{x}(\tilde{y})
	\end{dcases}.
\end{equation}
This local result allows to conclude that it is possible to move the output reference signal in a neighborhood $\Gamma({y}^o)$ of $y^o$ and still guarantee that a feasible solution to the tracking problem exists.

\begin{figure}[bt]
	\centering
	\includegraphics[width=0.6 \linewidth]{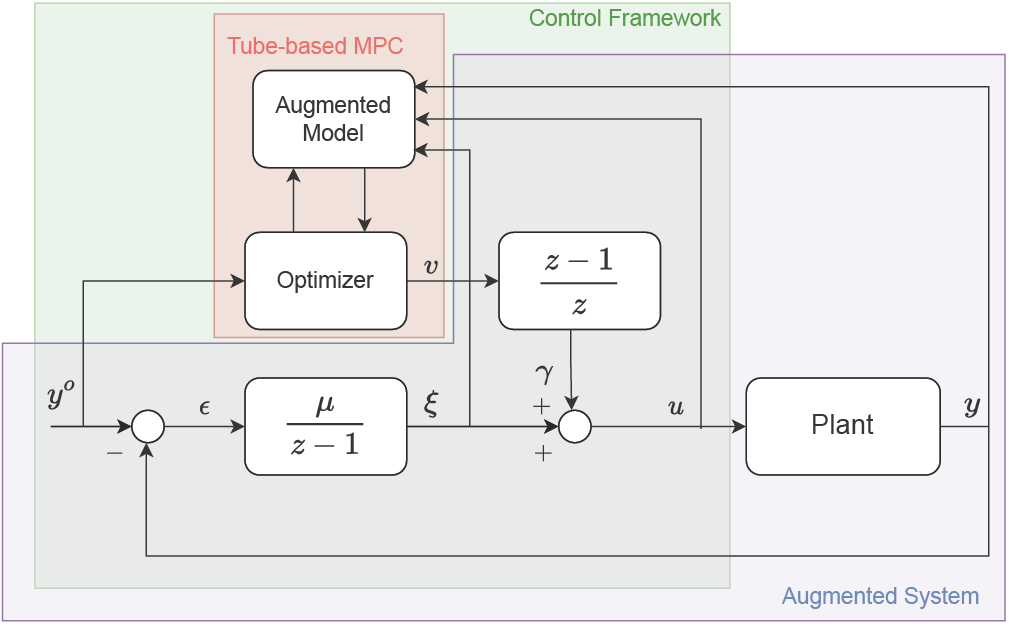}
	\caption{Schematic of the proposed control framework}
	\label{fig:controlscheme}
\end{figure}

\subsection{Control scheme}
The control scheme we propose is shown in Figure \ref{fig:controlscheme}. The regulator is made by three main blocks: an integrator of the output tracking error, an MPC algorithm, and a derivative action.
The rationale is the following:
\begin{itemize}
  \item The integrator acts on the output tracking error $\varepsilon_k = {y^o} - y_k$, so that, in light of the Internal Model Principle\REV{\cite{francis1976internal}, it is ensured that $\varepsilon_k$ asymptotically converges to zero also in the presence of asymptotically constant perturbations.} The gain $\mu$ of the integrator must be selected to guarantee stability properties. This choice is later discussed.
  \item The MPC algorithm provides improved performances in transient conditions and the fulfillment of input and output constraints, while at the steady state, its action is null due to the derivative action acting on its output.
  \item The derivative action allows to achieve stability results for the overall system by means of a standard zero terminal constraint formulation of MPC.
\end{itemize}
More formally, the integral block is described by
\begin{equation} \label{eq:control:integral}
	\xi_{k+1} = \xi_k + \mu ( {y^o} - \bar{y}_k ),
\end{equation}
with $\xi \in R^m$, $\mu \in R^{m,m}$, while the derivative action is
\begin{equation} \label{eq:control:derivative}
	\begin{dcases}
		\theta_{k+1} = v_k \\
		\gamma_k = v_k - \theta_k
	\end{dcases},
\end{equation}
where $v \in R^m$ is the output of the MPC regulator.
As apparent from Figure \ref{fig:controlscheme}, the resulting control action $u_k$ is given by
\begin{equation} \label{eq:control:u}
	u_k = \xi_k + \gamma_k.
\end{equation}
Note that, as already discussed, at steady state $\gamma= 0$, the control variable $u_k$ is uniquely defined by the integral action.

In view of the previous definitions, the overall system to be considered in the design of the MPC algorithm in nominal conditions is obtained from \eqref{eq:nnarx:compact},  \eqref{eq:control:integral}, \eqref{eq:control:derivative}, \eqref{eq:control:u}, and takes the form
\begin{equation}\label{eq:systemdynamics}
	\begin{dcases}
		\bar{x}_{k+1} = f(\bar{x}_k, u_k) \\
		\xi_{k+1} = \xi_k + \mu ({y^o} - \bar{y}_k) \\
		\theta_{k+1} = v_k \\
		\gamma_k = v_k -\theta_k \\
		u_k = \xi_k + \gamma_k \\
		\bar{y}_k = C \bar{x}_k
	\end{dcases},
\end{equation}
Defining the augmented state $\bar{\chi}_k = [\bar{x}_k^\prime, \xi_k^\prime, \theta_k^\prime ]^\prime$ and the augmented output $\bar{\zeta}_k = [\bar{y}_k^\prime, u_k^\prime]^\prime$, \eqref{eq:systemdynamics} can be rewritten in the general form
\begin{equation} \label{eq:control:enlarged_compact}
	\mathcal{\bar{S}}_a : \begin{dcases}
		\bar{\chi}_{k+1} = f_a(\bar{\chi}_k, v_k, {y^o}) \\
		\bar{\zeta}_k = g_a(\bar{\chi}_k, v_k)
	\end{dcases},
\end{equation}
with suitable functions $f_a$, $g_a$.

Let us now denote by $(\bar{\chi}(y^o),\bar{v}(y^o),\bar{\zeta}(y^o))$ the (feasible) state, input, and output equilibrium values corresponding to the output $\bar{y}=y^o$ of the enlarged system $\mathcal{\bar{S}}_a$.
Note that at equilibrium $\bar{\gamma}=0$, $\bar{\xi}=\bar{u}$, and $\bar{v}$ can take any value.
The following proposition can thus be stated, providing guidelines for the design of the integrator gain $\mu$. 

\begin{proposition}\label{prop:integrator}
Under \REV{Assumption \ref{ass:linearized_gs} and Assumption \ref{ass:linearized}}, there exists $\check{\mu} > 0$ such that, for any $\hat{\mu} \in (0, \check{\mu})$, the integrator gain
	\begin{equation} \label{eq:control:int_gain}
		\mu = \hat{\mu} \left [C (I - A_\delta)^{-1} B_\delta \right ]^{-1}
	\end{equation}
ensures that the enlarged system \eqref{eq:control:enlarged_compact} is \REV{locally asymptotically stable around the target equilibrium} $(\bar{\chi}(y^o), \bar{v}(y^o), \bar{\zeta}(y^o))$.
\end{proposition}
\begin{proof}
	\REV{In view of the Schur stability of $A_\delta$, as well as the reachability, observability and absence of derivative actions, Theorem 4 of Reference \citenum{scattolini1985parameter} allows to prove the claim.}
\end{proof}

Note that Proposition \ref{prop:integrator} not only guarantees the existence of a stabilizing integrator gain, but also provides a way to compute it. 
Indeed, $\check{\mu}$ can be estimated numerically, and $\mu$ can be selected via \eqref{eq:control:int_gain} by choosing $\hat{\mu} \in (0, \check{\mu})$.

\subsection{Nominal MPC design}\label{subsec:nominalmpc}
A nominal stabilizing nonlinear MPC law is now designed for the nominal augmented system model. The constrained optimization problem to be solved at any time instant $k$ is
\begin{subequations} \label{eq:nomMPC}
	\begin{align}
		\min_{v_{0|k}, ..., v_{N_p-1|k}} &  \sum_{i=0}^{N_p} \Big[ \big\| \bar{\chi}_{i|k} - \bar{\chi}(y^o) \big\|_Q^2 + \big\| \bar{\zeta}_{i|k} - \bar{\zeta}(y^o) \big\|_R^2 \Big]  \label{eq:nomMPC:cost} \\
		\text{s.t.} \quad & \forall i \in \{ 0, ..., N_p-1 \} \nonumber \\
		&  \bar{\chi}_{0|k} =\bar{\chi}_{k}  \label{eq:nomMPC:x0} \\
		& \bar{\chi}_{i+1|k} = f_a(\bar{\chi}_{i|k}, v_{i|k}, {y^o}) \label{eq:nomMPC:dynamics} \\
		& \bar{\zeta}_{i|k} = g_a(\bar{\chi}_{i|k},v_{i|k}) \label{eq:nomMPC:output} \\
		& S_x\bar{\chi}_{i|k} \in \mathcal{X} \label{eq:nomMPC:state} \\
		& S_u \bar{\zeta}_{i|k} \in \mathcal{U} \label{eq:nomMPC:actuator} \\
		&  \bar{\chi}_{N_p|k} = \bar{\chi}(y^o)\label{eq:nomMPC:terminal}
	\end{align}
\end{subequations}
In the MPC cost function \eqref{eq:nomMPC:cost}, $N_p$ denotes the prediction horizon, and the weight matrix $Q = \text{diag}(Q_x, Q_\xi, Q_\theta)$, where $ \text{diag}(\cdot)$ denotes the block-diagonal operator, penalizes the displacement of the state vector $\bar{\chi}$ from its equilibrium. The weight matrix $R$ is defined as $\text{diag}(R_e,R_u)$, where $R_e$ and $R_u$ penalize the output error and the control effort, respectively.

The augmented nominal model $\mathcal{\bar{S}}_a$, defined in \eqref{eq:control:enlarged_compact}, is used as predictive model, see \eqref{eq:nomMPC:dynamics} and \eqref{eq:nomMPC:output}, and it is initialized with known current values, see \eqref{eq:nomMPC:x0}.
Moreover, \eqref{eq:nomMPC:state} and \eqref{eq:nomMPC:actuator} enforce the saturation constraints on $\bar{x}_{i|k}$ and $u_{i|k}$, where the matrix $S_x$ selects $\bar{x}_{i|k}$ from the state vector $\xi_{i|k}$ and $S_u$ selects $u_{i|k}$ from the output vector $\bar{\zeta}_{i|k}$. Finally the terminal equality constraint \eqref{eq:nomMPC:terminal} is introduced to ensure \REV{the nominal closed-loop stability and recursive feasibility, according to well known arguments\cite{rawlings2017model}. Of course, other strategies guaranteeing the nominal closed-loop stability and recursive feasibility, such as those relying on the definition of a terminal set and a terminal cost\cite{magni2001stabilizing}, can also be adopted: at the price of a slightly more involved design phase, they can yield less conservative control actions\cite{magni2001stabilizing}.}

The solution to problem \eqref{eq:nomMPC} at the time instant $k$ yields the optimal control sequence $\boldsymbol{v}^*_k = \{ v^*_{0|k}, v^*_{1|k}, ..., v^*_{N_p-1|k}\}$. Then, according to the Receding Horizon approach,  only the first element in the sequence is applied and the implicit MPC control law reads
\begin{equation}
    v_k = \kappa_k(\bar{\chi}_k) = v^*_{0|k}.
\end{equation}

The procedure is repeated at the successive time step $k+1$, based on the measured state $\bar{\chi}_{k+1} = [\bar{x}_{k+1}^\prime, \xi_{k+1}^\prime, \theta_{k+1}^\prime ]^\prime$, which leads to a state-feedback control law. \smallskip

\subsection{Robust MPC design}
\label{subsec:tubempc}
The robust MPC controller is designed according to the popular tube-based approach, see Reference \citenum{MAYNE2005robust}. To this end, we compute for the perturbed error system \eqref{eq:errordynamics} a
Robust Positively Invariant (RPI) set $\Omega_x$ such that, if the real system state ${x}_k \in \bar{x}_k \oplus \Omega_x$, then ${x}_{k+i} \in  \bar{x}_{k+i} \oplus \Omega_x$ for all $i>0$, and for any admissible disturbance realization $\{ w_k, ..., w_{k+i-1} \}$.

\REV{In our case, in view of the structure of the state vector $x_k$, consistent with the autoregressive structure of the model,} the system matrix $A$ in \eqref{eq:nnarx:statespace} is Schur stable and nilpotent so that the RPI set is given by\cite{kolmanovsky1996invariantset}

\begin{equation}\label{eq:positive_invariant}
    \Omega_x = \sum_{j=0}^{N-1} A^j\mathcal{W},
\end{equation}
where $\sum_j$ is used to denote the Minkowski set addition and $N$ is the number of past input-output data of the NNARX model.
It is remarkable that, owing to the specific structure of NNARX models, the most critical issue in the design of robust tube-based MPC laws for nonlinear systems, i.e. the computation of $\Omega_x$, can be easily overcome.

In view of the previous considerations, the adopted formulation for nonlinear robust MPC reads as
\begin{subequations} \label{eq:tubeMPC}
	\begin{align}
		\min_{\substack{v_{0|k}, ..., v_{N_p-1|k} \\ \bar{x}_{0|k}}} &
		\sum_{i=0}^{N_p} \Big[ \big\| \bar{\chi}_{i|k} - \bar{\chi}(y^o) \big\|_Q^2 + \big\| \bar{\zeta}_{i|k} - \bar{\zeta}(y^o) \big\|_R^2 \Big]
		\label{eq:tubeMPC:cost} \\
		\text{s.t.} \quad & \forall i \in \{ 0, ..., N_p-1 \} \nonumber \\
		&  S_x \bar{\chi}_{0|k} \in \{ \bar{x}_{k} \} \oplus \Omega_x \label{eq:tubeMPC:x0} \\
		&  S_{\xi}\bar{\chi}_{0|k} =\xi_{k} \label{eq:tubeMPC:xi0} \\
		&  S_{\theta} \bar{\chi}_{0|k} =v_{k}\label{eq:tubeMPC:theta0} \\
		& \bar{\chi}_{i+1|k} = f_a(\bar{\chi}_{i|k}, v_{i|k}, y^o) \label{eq:tubeMPC:dynamics} \\
		& \bar{\zeta}_{i|k} = g_a(\bar{\chi}_{i|k},v_{i|k}) \label{eq:tubeMPC:output} \\
		& S_x \bar{\chi}_{i|k} \in \mathcal{X} \ominus \Omega_x \label{eq:tubeMPC:state} \\
		& S_u \bar{\zeta}_{i|k} \in \mathcal{U} \label{eq:tubeMPC:actuator} \\
		& \bar{\chi}_{N_p|k} = \bar{\chi}(y^o) \label{eq:tubeMPC:terminal_x}
	\end{align}
\end{subequations}

\begin{algorithm}[t]
\caption{Proposed robust offset-free MPC}\label{algo:mpcscheme}
\begin{algorithmic}[]
\State \textbf{Offline Phase:} 
\State train a ($\delta$ISS) NNARX model using the collected data 
\State select the integrator gain $\mu$ according to \eqref{eq:control:int_gain}
\State compute the set $\Omega_x$ according to \eqref{eq:positive_invariant}
\label{phaseoneendwhile}

\\
\State \textbf{Online Phase:} 
\For{each time step $k$}
\State obtain the current state $\bar{\chi}_k$
\State compute the equilibrium $\bar{\chi}(y^o)$ according to \eqref{eq:equilibrium}
\State solve the robust MPC problem \eqref{eq:tubeMPC} 
\State compute $u_k$ according to \eqref{eq:control:u}
\State apply $u_k$ to the  plant
\EndFor\label{phasetwoendwhile}   
\end{algorithmic}
\end{algorithm}

The cost function \eqref{eq:tubeMPC:cost} adopted is the same as in the nominal MPC, see \eqref{eq:nomMPC:cost}, and penalizes the distances of the augmented system's state and output from their respective targets.
The main difference lies, however, in the initialization of the component $\bar{x}_{0|k}$ of the augmented state vector.
Rather than fixing it to the known measured value $\bar{x}_k$, it is considered as a free optimization variable lying in the RPI set $ \bar{x}_{k} \oplus \Omega_x$, see \eqref{eq:tubeMPC:x0}.
Note that $S_x$, $S_\xi$, and $S_\theta$ denote the selection matrices that extract $\bar{x}_{i|k}$, $\bar{\xi}_{i|k}$, and $\bar{\theta}_{i|k}$ from $\bar{\chi}_{i|k}$, respectively.

The nominal augmented model $\mathcal{\bar{S}}_a$ is then used as predictive model, see \eqref{eq:tubeMPC:dynamics} and \eqref{eq:tubeMPC:output}.
To ensure the robust constraint satisfaction, the constraint on state $\bar{x}_{i|k}$ is tightened, see \eqref{eq:tubeMPC:state}.
The input variable $u_{i|k}$ is constrained to fulfill the actuator constraints, see \eqref{eq:tubeMPC:actuator}.

Solving the optimization problem \eqref{eq:tubeMPC}, the optimal control sequence $\boldsymbol{v}_k^\star = \{ v_{0|k}^\star, ..., v_{N_p-1|k}^\star \}$ is retrieved.
Then, according to the Receding Horizon principle, the first optimal control action $v_{0|k}^\star$ is applied, and at the following time-step the entire procedure is repeated.
The corresponding implicit robust control law is denoted by
\begin{equation} \label{eq:tube:contrllaw}
	v_k = \kappa_k(\bar{\chi}_k) = v_{0|k}^\star
\end{equation}
Notably, such control law keeps the trajectory of the state  of the perturbed system \eqref{eq:nnarx:real}, i.e., $x_k$, in the robust control invariant set $\Omega_x$ centered along the nominal state trajectory $\bar{x}_k$, which implies that the output $y_k$ of the perturbed system is constrained in the set $\Omega_y = C \Omega_x$, centered around the nominal output $\bar{y}_k$.
Moreover, \REV{if the plant states converge to constant values, the output tracking error is guaranteed to be asymptotically null by the Internal Model Principle\cite{francis1976internal}, i.e., robust offset-free tracking is attained.}

Note that, as proved in Reference \citenum{MAYNE2005robust}, the RPI set $\Omega_x$ (centered around the nominal state trajectory) is robustly exponentially stable for the controlled uncertain system \eqref{eq:nnarx:real}. 
Therefore, recursive feasibility is guaranteed by Proposition 3.14 of Reference \citenum{rawlings2017model}.
\smallskip

\REV{The overall process of the proposed robust offset-free MPC is summarized in Algorithm \ref{algo:mpcscheme}.}

\section{Numerical Example} \label{sec:example}
\subsection{Benchmark system}\label{sec:example:plant}
 A water-heating benchmark system depicted in Figure \ref{fig:plant} is used to test the proposed control framework. The goal of the system is to regulate the temperature of the outlet water $T$ to a desired value with the required flow rate. The inlet water flow rate is assumed to be the same as the outlet water flow rate $w$, keeping the water level in the tank constant. $w_c$ denotes the inlet gas flow rate, which is burnt to heat the metal, and subsequently, heat the water in the tank. $T_i$ and $T_m$ denote the temperature of the inlet water and the metal respectively.  
The system dynamics, derived from the energy balance equations, are described by
\begin{equation}\label{eq:example:plant}
\begin{dcases}
     \dot{T} =\frac{1}{\rho_w A_t z_w}\left [ w \left ( T_{i} - T \right )+\frac{k_{lm} A_t}{c_w}\left ( T_{m }- T \right ) \right ], \\
    \dot{T}_m = \frac{1}{M_{m} c_{m}}\left [ -k_{lm} A_t \left ( T_{m} - T \right )+\sigma k_{f} w_{c} \left ( T_{f}^{4}-T_{m}^{4} \right ) \right ],
\end{dcases}
\end{equation}
where the values of the parameters appearing in \eqref{eq:example:plant} are reported in Table \ref{tab:system_parameters}.

The model has one manipulable input $u=w_c$ , one output $y=T$ and two system states $x= [T,T_m]^{'}$. Both $T_i$ and $w$ in the system dynamics \eqref{eq:example:plant} can be treated as disturbances, i.e., $d = [T_i,w]^{'}$. The nominal values of both disturbances are reported in Table \ref{tab:system_parameters}, together with other parameters. Moreover, there is a constraint on the input,
\begin{equation}\label{eq:input:region}
    w_c \in [0.05, 0.18].
\end{equation}

\begin{figure}[bt]
	\centering
	\includegraphics[width=250pt,height=200pt]{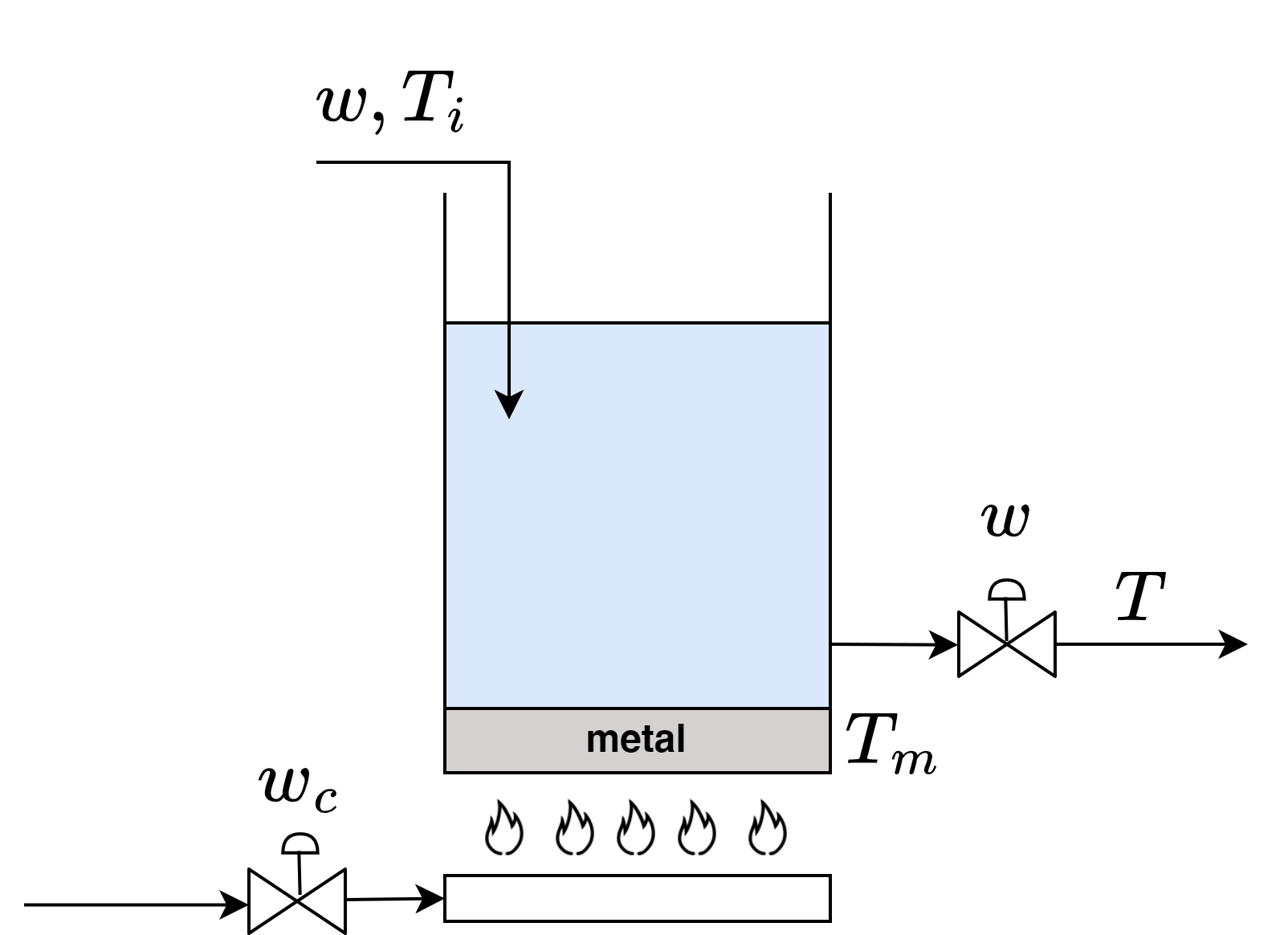}
	\caption{Water-heating system illustration}
	\label{fig:plant}
\end{figure}
The benchmark system has been implemented in Simulink. The simulations have been carried out to collect the input-output data, which will be used to train Neural NARX model and test the proposed control algorithm.
\begin{table}[t]
\centering
\caption{Benchmark system parameters}
\label{tab:system_parameters}
\begin{tabular*}{350pt}{@{\extracolsep\fill}cccc@{\extracolsep\fill}}
\toprule
\textbf{Parameter}  & \textbf{Description} & \textbf{Value}  & \textbf{Units} \\ \midrule
$A_t$ & Tank's cross-section &$\frac{\pi}{4}$ & $m^2$ \\
$\rho_{w}$ & Water's density & $997.8 $ & $\frac{kg}{m^3}$\\
$c_w$ & Water's specific heat & $4180.0$ & $\frac{J}{kg \cdot K}$\\
$M_{m}$ & Metal plate's mass & $617.32$ & $kg$\\
$c_m$ & Metal's specific heat & $481.0$ & $\frac{J}{kg \cdot K}$\\
$\sigma$ & Radiation coefficient & $5.67\times 10^{-8}$ & $\frac{W}{m^2 \cdot K^4}$ \\
$k_{lm}$ & Heat exchange coefficient & $3326.4$ & $\frac{kg}{s^3 \cdot K}$ \\
$T_f$ & Flame's temperature & $1200$ & $K$\\
$k_f$ & Heat exchange coefficient & $8.0$ & $\frac{m^2 \cdot s}{kg}$\\
$z_w$ & Water level & $2.0$ & $m$ \\
$\bar{w}$ & Nominal water flow rate & $1.0$ & $\frac{kg}{s}$ \\
$\bar{T}_i$ & Nominal inlet water temperature & $298$ & $K$ \\
\bottomrule
\end{tabular*}
\end{table}

\subsection{NNARX model training}
To generate the dataset for training the Neural NARX model of the plant \eqref{eq:example:plant}, the simulator has been fed with a Multilevel Pseudo-Random Signal (MPRS) to properly excite the system in the operating region \eqref{eq:input:region}.
One input-output trajectory $\boldsymbol{T}_{exp}$, with a total length of $2500$ time steps, has been collected with sampling time $\tau_s = 120 s$.
According to the Truncated Back-Propagation Through Time (TBPTT) principle\cite{jaeger2002tbptt}, $N_{t} = 120$ subsequences of length $T_s^{'} = 400$ time steps, have been extracted from $\boldsymbol{T}_{exp}$. We denote these sequences $(\boldsymbol{u}^{\{i\}}, \boldsymbol{y}^{\{i\}})$ with $i \in \mathcal{I}_t = \{ 1, ..., N_t \}$.
The validation set and test set consist of  $N_v=30$ and $N_f=1$ subsequences respectively, with length $T_s = 1000$ time steps. Both sets are constructed by using completely independent simulation data.
We denote validation and test set by the set of indices $\mathcal{I}_v = \{ N_t + 1, ..., N_t + N_v \}$ and $\mathcal{I}_f = \{N_t + N_v + 1, ..., N_s \}$, respectively,  where $N_s = N_t + N_v + N_f$.

The training procedure has been carried out with PyTorch 1.9 and Python 3.9.
The adopted Neural NARX model features a single-layer ($M=1$) FFNN  with $30$ neurons and activation function $\sigma_t = \tanh$ in \eqref{eq:model:ffnn2}. The look-back horizon is $N=5$. During the training procedure, the Mean Square Error (MSE) over the training set $\mathcal{I}_t$ is minimized. \REV{ A suitable regularization term is included in the cost function \cite{bonassi2021nnarx} so that the model enjoys the $\delta$ISS property}. The training takes 1288 epochs until the modeling performance over the validation set $\mathcal{I}_v$ saturates.
\begin{figure}[t]
	\centering
	\includegraphics[width=0.75\columnwidth]{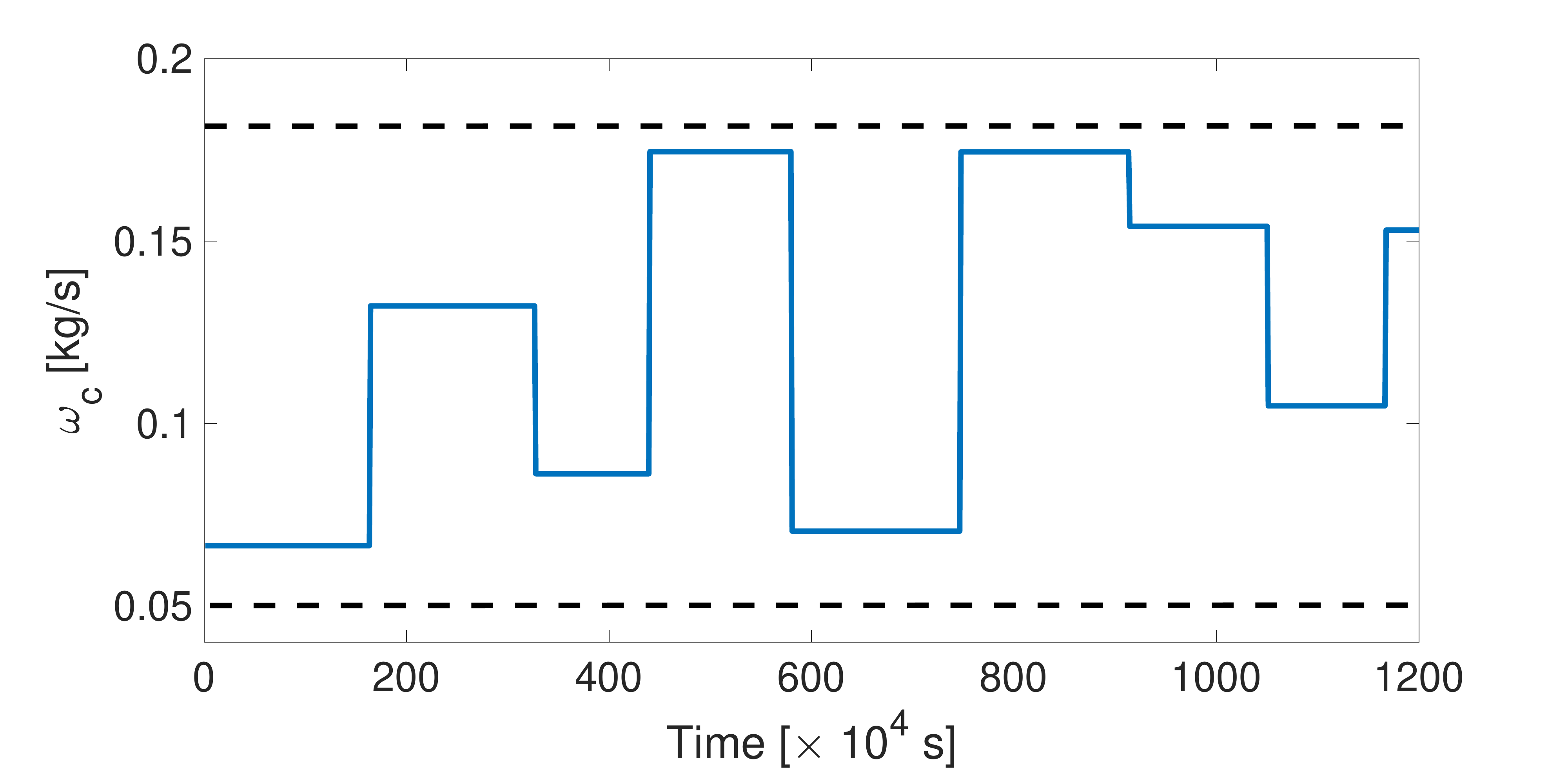}
	\caption{Input sequence used for validation.}
	\label{fig:test_input}
	\hspace{0.5cm}
	\includegraphics[width=0.75\columnwidth]{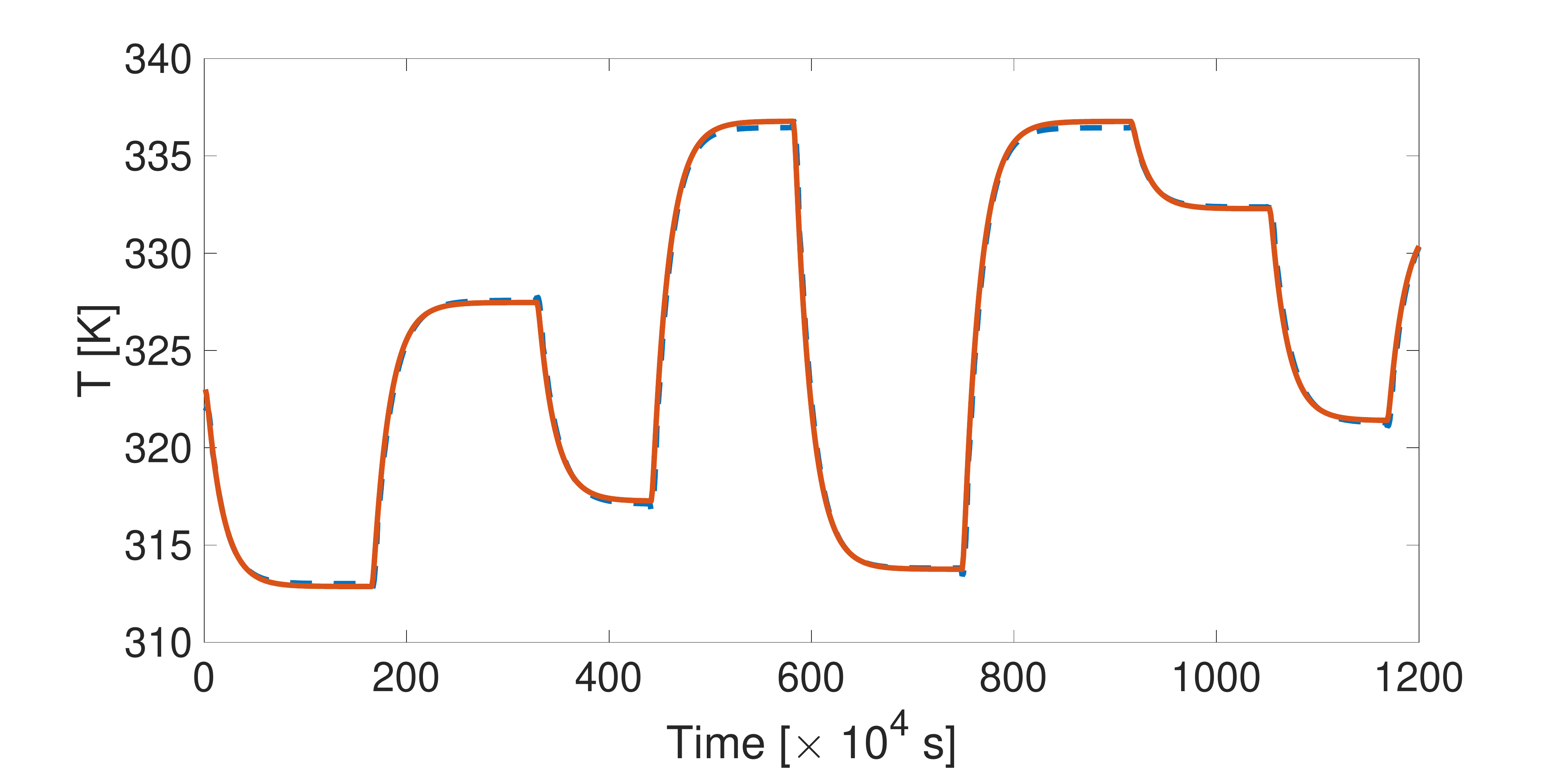}
	\caption{Open-loop prediction (blue dashed line) versus ground truth (red continuous line).}
	\label{fig:test_output}

\end{figure}
Lastly, the trained Neural NARX has been tested on the independent test set $\mathcal{I}_f$. In Figure \ref{fig:test_input}, the test input sequence $\boldsymbol{u}^{\{ts\}}$ is shown. In Figure \ref{fig:test_output}, the resulting output sequence $\boldsymbol{y}^{\{ts\}} = \left \{ y^{\{ts\}}_0, y^{\{ts\}}_1,...,y_{T_s}^{\{ts\}}\right \}$ of the Neural NARX model is compared to the ground truth sequence $\bar{\boldsymbol{y}}^{\{ts\}} =\left\{ \bar{y}^{\{ts\}}_0,...,\bar{y}^{\{ts\}}_{T_s}\right\}$. To quantitatively evaluate the model performance, we introduce the FIT index, defined as
\begin{equation}
	\text{FIT} = 100 \left( 1 -  \frac{\sum_{k=0}^{T_s}\| y_k^{\{ts\}} - \bar{y}_{k}^{\{ts\}}\|_2}{\sum_{k=0}^{T_s}\| \bar{y}_{k}^{\{ts\}} - \bar{y}_{avg} \|_2} \right),
\end{equation}
where $\bar{y}_{avg}$ is the average over sequence $\bar{\boldsymbol{y}}^{\{ts\}}$.
The trained Neural NARX model has scored $92.8 \%$ FIT index, which entails satisfactory modeling performance.

\begin{figure}[t]
	\centering
	\includegraphics[width=0.75\columnwidth]{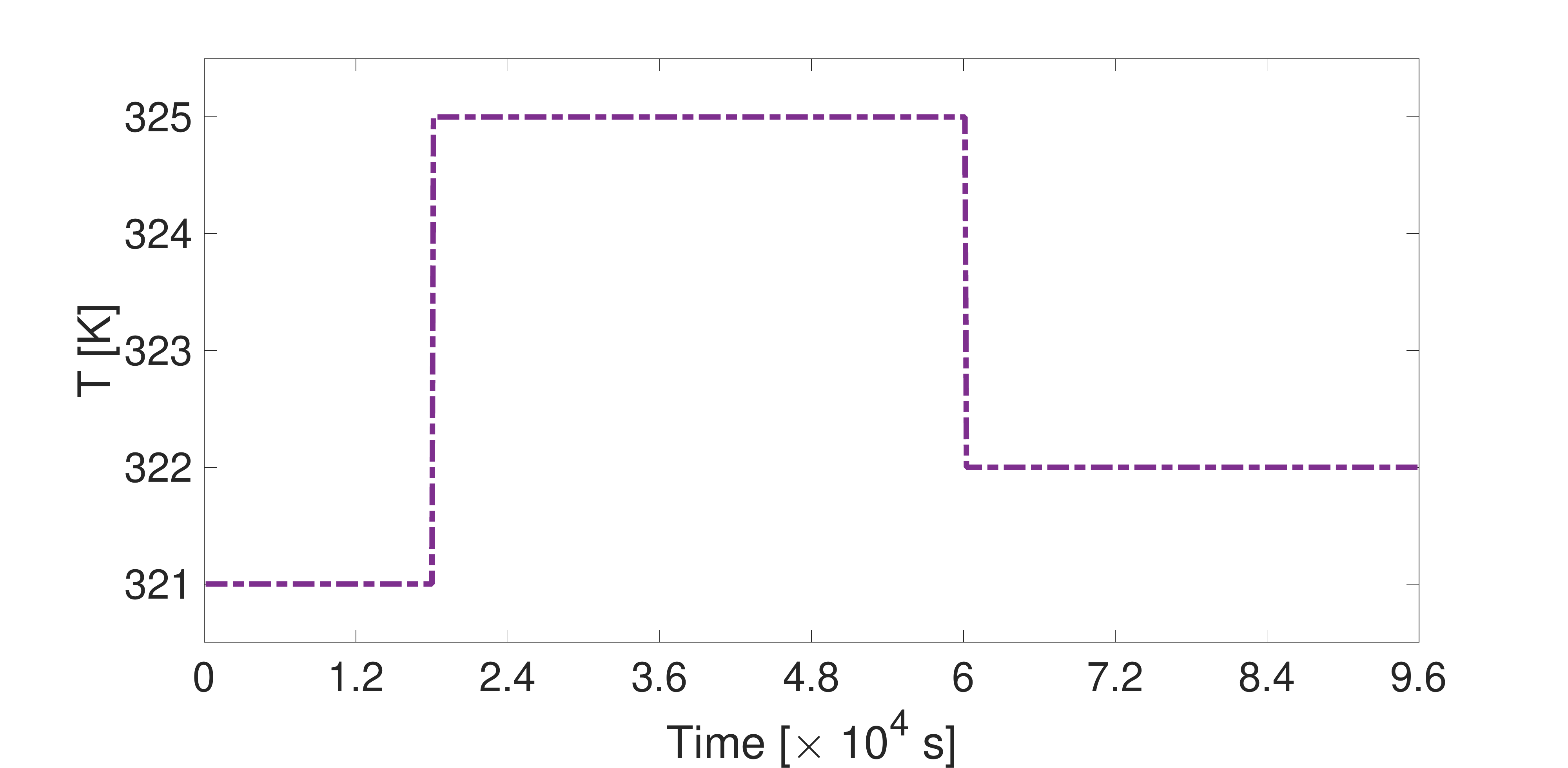}
	\caption{Piecewise-constant output reference trajectory.}
	\label{fig:reference}
\end{figure}

\subsection{Control synthesis}
The implementation of the control strategy described in Section \ref{sec:control} is presented in the following. 
The temperature reference signal that will be considered for assessing the closed-loop performances is depicted in Figure \ref{fig:reference}. 
\smallskip

\subsubsection*{Nominal MPC}

Let us first consider the nominal offset-free tracking scheme proposed in Section \ref{subsec:nominalmpc}. The goal of such control strategy is to track piecewise constant water temperature reference signals relying on the nominal model of the system.

The hyperparameters for MPC alogrithm are reported in Table \ref{tab:mpc_hyperparameters}. 
Note that $Q_\theta$ is chosen much smaller than $Q_x$ and $Q_\xi$, as the sole purpose of such weight matrix is to penalize large values of $\theta_k$.
The integrator gain $\mu = 0.14$ is selected according to Proposition \ref{prop:integrator}, as it ensures the local asymptotic stability of the augmented system around any equilibrium of interest.

It should be noted that, for every new reference point $y^o$ in Figure \ref{fig:reference}, the associated equilibrium point $({\bar{x}(y^o)},\bar{u}(y^o))$ should be computed again by solving \eqref{eq:equilibrium}, together with the target value for the augmented system $(\bar{\chi}(y^o),\bar{v}(y^o),\bar{\zeta}(y^o))$. 

\begin{figure}[t]
	\centering
	\includegraphics[width=0.75\columnwidth]{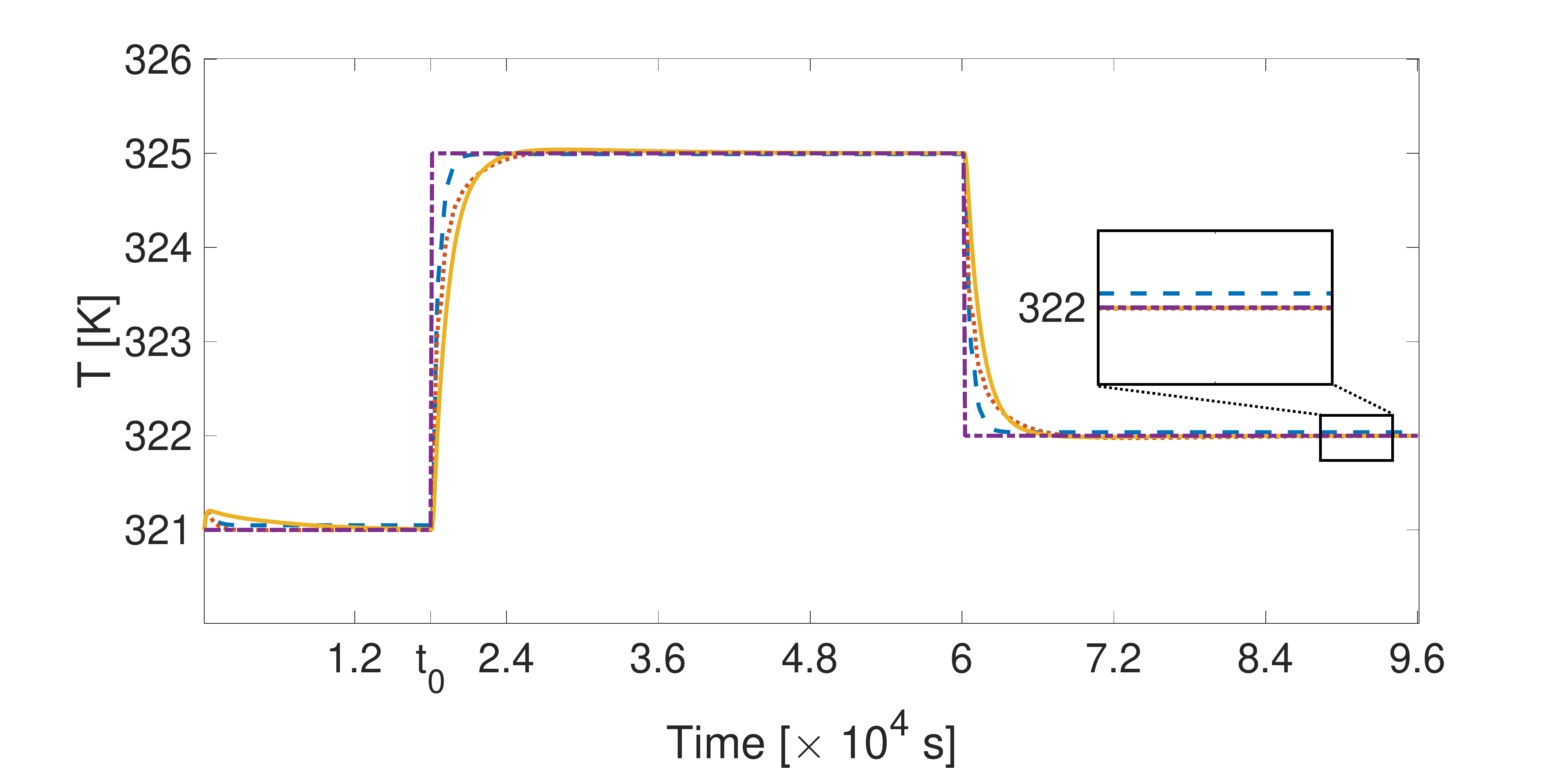}
	\caption{Closed-loop output tracking performances of the proposed nominal approach (red dotted line), proposed robust approach (yellow continuous line) compared to that of the DEB-MPC (blue dashed line). The reference is represented by the purple dashed-dotted line.
	}
	\label{fig:y_compare}
	\centering
	\includegraphics[width=0.75\columnwidth]{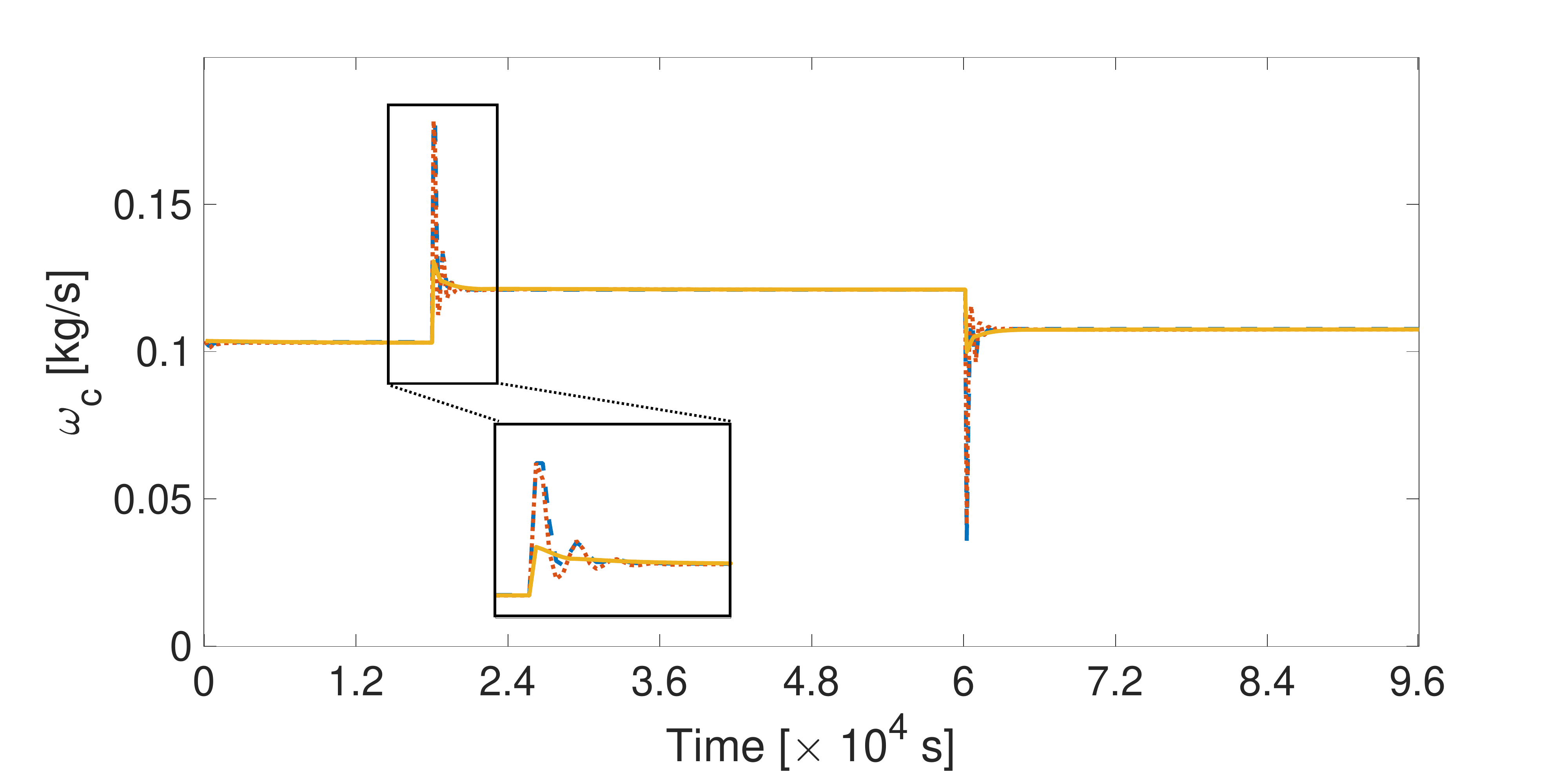}
	\caption{Control input of the proposed nominal approach (red dotted line), proposed robust approach (yellow continuous line) compared to that of the DEB-MPC (blue dashed line). }
	\label{fig:u_compare}
\end{figure}

\subsubsection*{Robust MPC}
Having developed the Nominal MPC architecture, let us now address the synthesis of the robust control scheme proposed in Section \ref{subsec:tubempc}. 
The first step is to estimate the upper bound of the uncertainty $w_k$, i.e. $\check{w}$, see \eqref{eq:mismatch:upperbound}. 
In the spirit of data-driven control approach, in this paper we estimate the upper bound from simulation data.

In view of \eqref{eq:perturbation_def} and \eqref{eq:nnarx:statespace:matrices}, $\check{w}$ can be retrieved as the maximum mismatch between the NNARX open-loop output simulation, i.e. $\bar{y}_k$, and the corresponding true plant's output, i.e. $y_k$.
To this end, the plant \eqref{eq:example:plant} has been simulated with $30$ random input trajectories, and the maximum difference between the corresponding measured output sequences and the NNARX open-loop simulation is computed, yielding $\check{w} = 0.031$.
The set $\Omega_x$ is finally computed via \eqref{eq:positive_invariant}.
\smallskip

\subsubsection*{Comparison and Discussion}
The popular offset-free Disturbance Estimation Based MPC\cite{morari2012nonlinear} (DEB-MPC) has been implemented as baseline to which compare the proposed control architecture.
This standard offset-free control scheme requires the augmentation of the NNARX system model with a fictitious disturbance on the output, which is customarily assumed to be constant.
Such disturbance is then estimated via suitably designed state observer: here a Moving Horizon Estimator (MHE) is employed.
For the sake of a fair comparison, the DEB-MPC design parameters (i.e., the prediction horizon and weights) are chosen in line with that reported in Table \ref{tab:mpc_hyperparameters}.

The closed-loop tracking performances of the proposed nominal and robust approaches, together with that of the baseline DEB-MPC regulator are reported in Figure \ref{fig:y_compare}. 
Albeit the DEB-MPC approach achieves better performance during transients, it is unable to attain asymptotic offset-free tracking, presumably due to the inadequacy of the disturbance model, whose accuracy is known to be paramount to attain the desired closed-loop performances \cite{tatjewski2020algorithms}.
On the other hand, \REV{both the nominal MPC and the robust MPC attain asymptotic zero-error, owing to the presence of integral action}. 

In Figure \ref{fig:u_compare}, the control actions requested by the three approaches are depicted.
It is apparent that all of them allow to fulfil input saturation constraints.
It can also be noted that the proposed robust approach requires the most moderate control action compared to the other two approaches, thanks to the tube-based design.

\begin{table}[t]
\centering
\caption{Hyperparameters for MPC algorithm}
\label{tab:mpc_hyperparameters}
\begin{tabular}{cc|cc}
\toprule
\textbf{Parameter}  & \textbf{Adopted Value}& \textbf{Parameter}  & \textbf{Adopted Value} \\ \midrule
$N_p$ & 50&$Q_x$ &$\text{diag}(R,R,R,R,R)$\\
$R_e$ & 10&$Q_\xi$ & 1\\
$R_u$ & 0.1&$Q_\theta$ & $10^{-5}$\\
\bottomrule
\end{tabular}
\end{table}

\REV{The nonlinear optimization problems in this numerical example have been solved with CasADi \cite{Andersson2019}: the empirical distribution of the computational time is depicted in Figure  \ref{fig:opt_time}.
Note that most of the instances of the optimization problem are solved in less than $2.5~s$, well below the sampling time of $120~s$.}

In order to show the robustness of the tube-based MPC approach, the closed-loop output trajectory starting from $t_0 = 1.8 \times 10^4 s$ is shown in Figure \ref{fig:y_tube_predict}, where $t_0$ denotes the time at which the setpoint is changed from $321 K$ to $325 K$, see Figure \ref{fig:reference}. 
The nominal output trajectory predicted at $t_0$ using the optimization in \eqref{eq:tubeMPC} is also shown.
\REV{
It is clear that the tube-based MPC allows to maintain the closed-loop output trajectory within the tube $\Omega_y$ centered around the nominal output, in spite of the plant-model mismatch.}

\begin{figure}[t]
\centering
\includegraphics[width=0.8\columnwidth]{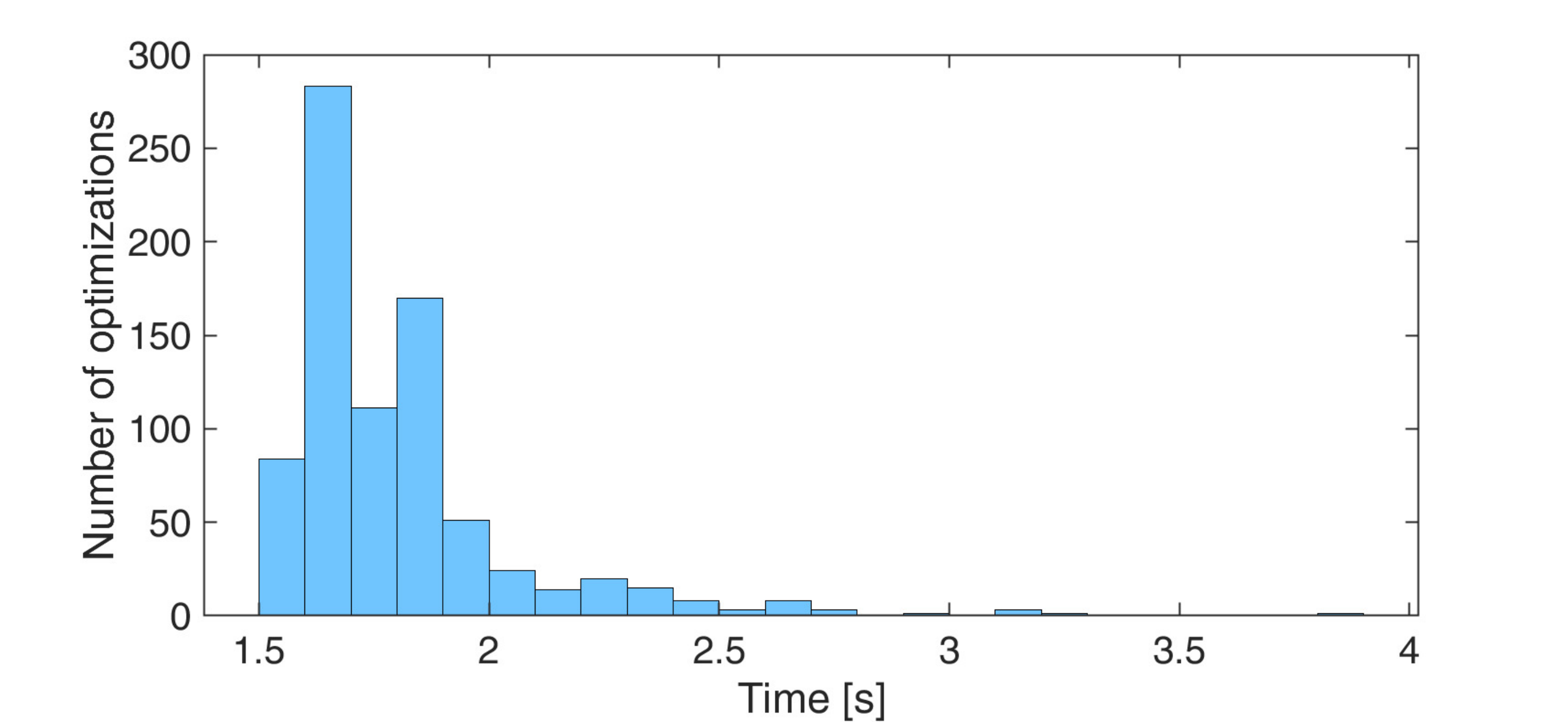}
	\caption{Distribution of the time required for solving optimization problems
	}
	\label{fig:opt_time}
\end{figure}

\begin{figure}[t]
	\centering
\includegraphics[width=0.8\columnwidth]{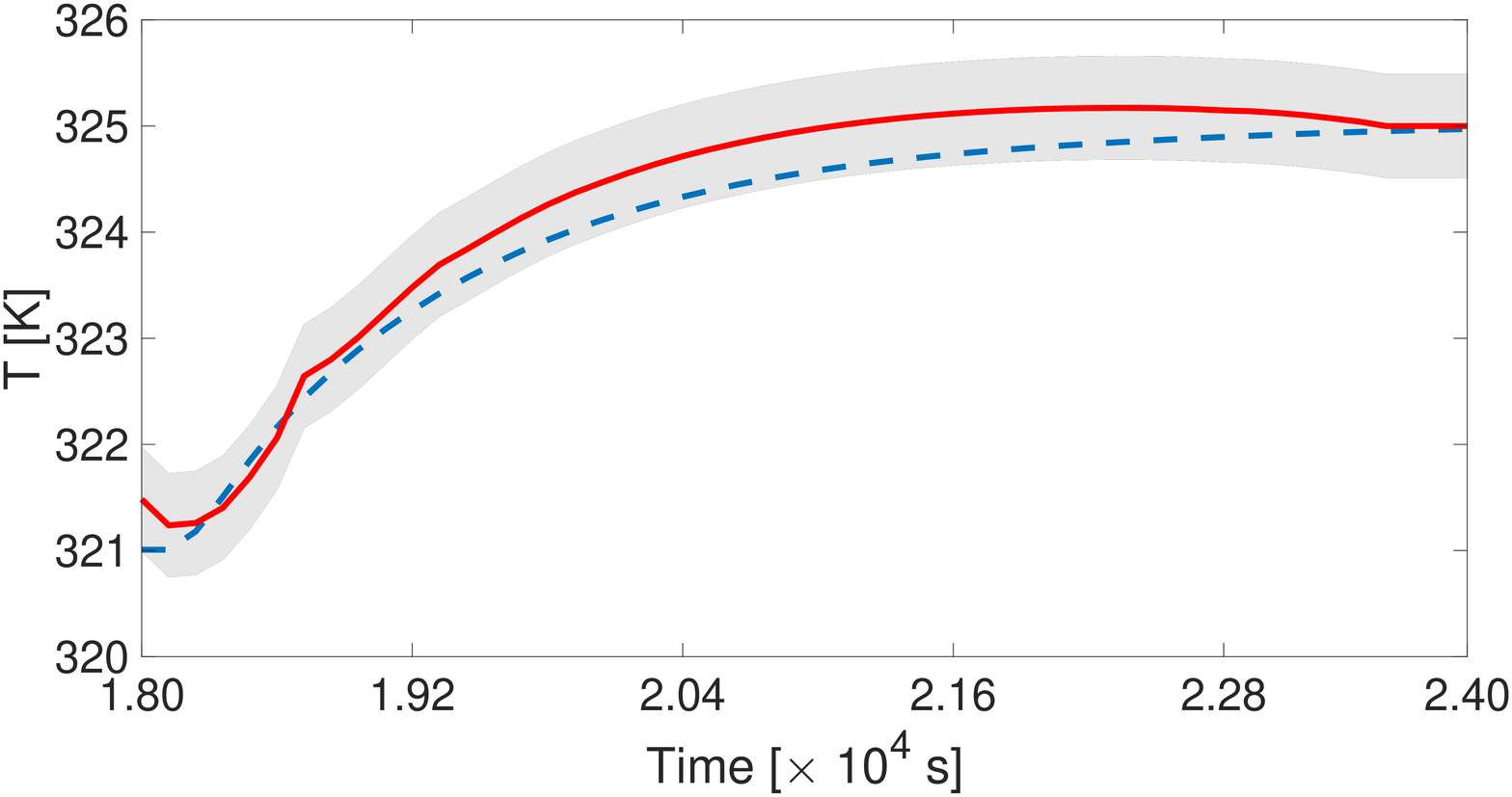}
	\caption{Closed-loop output trajectory of robust approach (blue dashed line) compared to  the open-loop nominal output trajectory (red continuous line) predicted at time $t_0$. The shaded area is the tube around the nominal trajectory.
	}
	\label{fig:y_tube_predict}
\end{figure}

It is worth noticing, however, that the uncertainty bound $\check{w}$ is rather conservative, since the maximum error only occurs when the setpoint change, whereas it is significantly smaller at steady state.
Time-varying tubes\cite{gonzalez2011online} could therefore represent a more efficient trade-off between conservativeness and robustness. 
Future research efforts will therefore be devoted to these approaches.

\section{Conclusions} \label{sec:conclusions}
In this paper, a nonlinear robust Model Predictive Control (MPC) scheme is proposed for system learned by Neural NARX models.
The model has been augmented with two additional elements, the integral action on the output tracking error and the derivative action on the MPC control variable. \REV{Moreover, a tube-based MPC is implemented, which keeps the real state within the prescribed tube around the nominal state and allows to guarantee robust constraint satisfaction when asymptotically constant modeling errors are present.}
The proposed control scheme features robust closed-loop stability and offset-free tracking capabilities.
Finally, the proposed control scheme has been tested on a water heating benchmark system, demonstrating its potentialities.

\section*{Acknowledgements}
\vspace{1mm}
\begin{minipage}[l]{0.2\columnwidth}
	\includegraphics[width=\textwidth]{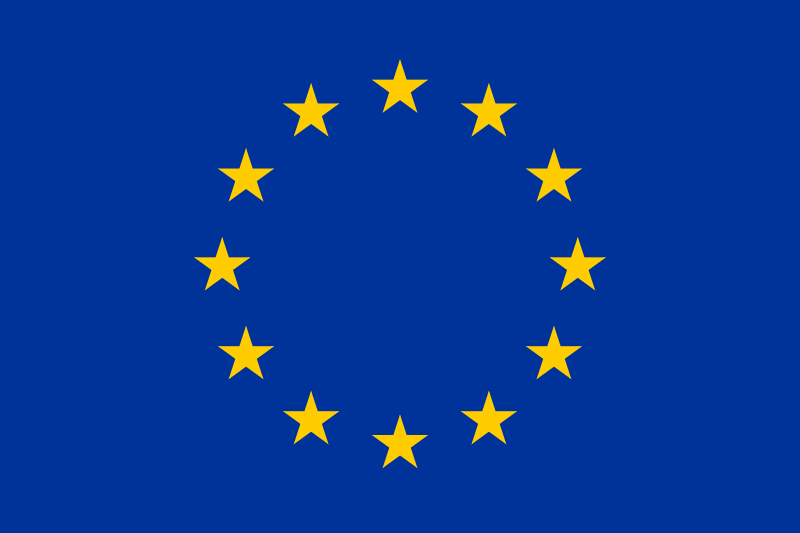}
	\label{fig:euflag}
\end{minipage}
\hspace{3mm}
\begin{minipage}[right]{0.7\columnwidth}
	This project has received funding from the European Union’s Horizon 2020 research and innovation programme under the Marie Skłodowska-Curie grant agreement No. 953348
\end{minipage}
\vspace{2mm}

\bibliography{wileyNJD-AMA.bib}%

\appendix
\section{Stability properties} \label{appendix:stability}
\REV{
In previous works\cite{terzi2019learning, bonassi2020stability, bonassi2021nnarx, bonassi2022survey}, the stability properties of the main RNNs have been investigated, with the aim of learning models safe and robust against input perturbations, that could facilitate the design of theoretically sound control laws. 
In this context, a very popular stability notion for nonlinear systems is that of the Incremental Input-to-State Stability ($\delta$ISS), whose definition is reported below.
\begin{definition}[$\delta$ISS] \label{def:deltaiss}
	 System \eqref{eq:nnarx:compact} is $\delta$ISS if there exist functions $\beta \in \mathcal{KL}$ and $\gamma \in \mathcal{K}_\infty$ such that, for any pair of initial states $\bar{x}_{a,0}$ and $\bar{x}_{b,0}$, and any pair of input sequences $\boldsymbol{u}_a \in \mathcal{U}$ and $\boldsymbol{u}_b \in \mathcal{U}$, at any $k \geq 0$ it holds that
	\begin{equation} \label{eq:delta_isseltaiss:definition}
		\| \bar{x}_{a,k} - \bar{x}_{b,k} \|_2 \leq \beta(\| \bar{x}_{a,0} - \bar{x}_{b,0} \|_2, k) + \gamma(\| \boldsymbol{u}_a - \boldsymbol{u}_b \|_{2,\infty}),
	\end{equation}
	 where $\bar{x}_{*,k}$ denotes the state trajectory of the system \eqref{eq:nnarx:compact}, initialized in $\bar{x}_{*,0}$ and fed by the sequence $\boldsymbol{u}_*$. 
\end{definition}}
\REV{
\begin{definition}[Exponential $\delta$ISS] \label{def:exponential_deltaiss}
	If system \eqref{eq:nnarx:compact} is $\delta$ISS in the sense specified by Definition \ref{def:deltaiss}, and the function $\beta$ takes an exponential form, i.e. there exist constants $\rho > 0$ and $\lambda \in (0, 1)$ such that $\beta(\| \bar{x}_{a,0} - \bar{x}_{b,0} \|_2, k) \leq \rho \| \bar{x}_{a,0} - \bar{x}_{b,0} \|_2 \, \lambda^k$, the system is said to be exponentially $\delta$ISS.
\end{definition}
In Reference \citenum{bonassi2021nnarx}, the $\delta$ISS property of NNARX models has been investigated, showing that there exists a sufficient condition on the set of weights, $\Phi$, guaranteeing its exponential $\delta$ISS.
In particular, such sufficient condition is the satisfaction of the following nonlinear inequality on the network's weights
\begin{equation}
	\prod_{l=0}^M \| U_l \|_2 - \frac{1}{\big( \prod_{l=1}^M L_{\sigma_l} \big) \sqrt{M}} <0,
\end{equation}
where $L_{\sigma_l}$ denotes the known Lipschitz constant of the $l$-th activation function.
Such sufficient condition can be enforced during the training procedure by means of a suitable regularization term, so that a provenly $\delta$ISS NNARX model is trained.}

\section{Proof of Proposition 1} \label{appendix:proof_stability}
\begin{proof}
Let $\delta x_k$ and $\delta u_k$ be the displacement from the equilibrium point $(\bar{x}, \bar{u})$, i.e. $x_k = \bar{x} + \delta x_k$ and $u_k = \bar{u} + \delta u_k$.
The nonlinear system \eqref{eq:nnarx:compact} can be rewritten as
\begin{equation} \label{eq:proof:delta_system_full}
	\delta x_{k+1} + \bar{x} = f(\bar{x} + \delta x_k, \bar{u} + \delta u_k).
\end{equation}
Since the goal is to analyze the asymptotic stability of the linearized system, for simplicity it is assumed that $\delta u_k = 0$.
It is worth noticing, however, that this proof could be easily extended to consider $\delta u_k \neq 0$, at the price of more involved computations.
Under this simplification, \eqref{eq:proof:delta_system} reads
\begin{equation} \label{eq:proof:delta_system}
	\delta x_{k+1} + \bar{x} = f(\bar{x} + \delta x_k, \bar{u}).
\end{equation}

System \eqref{eq:proof:delta_system} can be recast as its linearization plus the linerization error $\varepsilon$
\begin{equation}
	\delta x_{k+1} = A_\delta \delta x_k + \varepsilon(\delta x_k),
\end{equation}
where
\begin{equation}
	A_\delta = \left. \frac{\partial f}{\partial x_k} \right\lvert_{\bar{x}, \bar{u}}.
\end{equation}
The goal is to show that the linear system
\begin{equation}\label{eq:proof:linear}
    \delta x_{k+1} = A_\delta \delta x_k
\end{equation}
is asymptotically stable.
Along the lines of Reference \citenum{khalil2002nonlinear}, the linearization error is first bounded as follows.

Consider the $i$-th state component, with $i \in \{ 1, ..., n \}$. In light of the Mean Value Theorem there exists $\tilde{x}$ between $\bar{x}$ and $\bar{x} + \delta x_k$ such that
\begin{equation}
\begin{aligned}
&f_i(\bar{x} + \delta x_k, \bar{u}) - f_i(\bar{x}, \bar{u}) = \left. \frac{\partial f_i(x, \bar{u})}{\partial x} \right\lvert_{\tilde{x}, \bar{u}}  \delta x_k \\
& \,\, = \left.\frac{\partial f_i(x, \bar{u})}{\partial x} \right\lvert_{\bar{x}, \bar{u}} \!\! \delta x_k + \bigg[ \left. \frac{\partial f_i(x, \bar{u})}{\partial x} \right\lvert_{\tilde{x}, \bar{u}} \! - \! \left.\frac{\partial f_i(x, \bar{u})}{\partial x} \right\lvert_{\bar{x}, \bar{u}} \bigg] \delta x_k \\
& \,\, = A_{\delta i} \delta x_k + \tilde{\varepsilon}_i(\delta x_k) \delta x_k
\end{aligned}
\end{equation}
In the light of the assumption on the Lipschitz continuity of the gradient of $f(x,\bar{u})$, it holds that
\begin{equation}
    \begin{aligned}
        \| \tilde{\varepsilon}_i(\delta x_k) \|^2_2 &\leq \left\| \left. \frac{\partial f_i(x, \bar{u})}{\partial x} \right\lvert_{\tilde{x}, \bar{u}} \! - \! \left.\frac{\partial f_i(x, \bar{u})}{\partial x} \right\lvert_{\bar{x}, \bar{u}} \right\|^2_2 \\
        & \leq L_1^2 \| \tilde{x} - \bar{x} \|_2^2 \leq L_1^2 \| \delta x_k \|_2^2.
    \end{aligned}
\end{equation}
Hence, being $\varepsilon(\delta x_k) = \tilde{\varepsilon}(\delta x_k) \delta x_k$,
 the linearization error can be bounded as
\begin{equation} \label{eq:proof:linearization_err_bound}
    \begin{aligned}
        \| \varepsilon(\delta x_k) \|_2 &\leq \| \delta x_k \|_2 \, \sqrt{\sum_{i=1}^n \big\| \tilde{\varepsilon}_i(\delta x_k ) \big\|_2^2 \, \|\delta x_k \|^2_2}  \\
        &\leq L_\varepsilon \| \delta x_k \|_2^2,
    \end{aligned}
\end{equation}
where $L_\varepsilon = L_1 \sqrt{n}$.

At this stage, let us recall that the $\delta$ISS property  implies the \REV{exponential} Global Asymptotic Stability (GAS) of any equilibrium.
Indeed, recalling that $\delta u_k = 0$, from \eqref{eq:delta_isseltaiss:definition} it follows that
\begin{equation} \label{eq:proof:exponential_gas}
    \| \delta x_k \|_2 \leq \rho  \| \delta x_0 \|_2 \lambda^k.
\end{equation}
This allows to invoke Theorem 5.8 in Reference \citenum{bof2018lyapunov}, which, under the assumption of exponential GAS, guarantees the existence of a quadratic Lyapunov function $V(\delta x)$ for the nonlinear system \eqref{eq:proof:delta_system}.
That is, there exist positive constants $c_1$, $c_2$, $c_3$, $c_4$, such that
\begin{subequations}
    \begin{gather}
    c_1 \| \delta x_k \|_2^2 \leq V(\delta x_k) \leq c_2 \| \delta x_k \|_2^2, \label{eq:proof:lyapunov:bounds}\\
    V(A_\delta \delta x_{k} + \varepsilon(\delta x_k)) - V(\delta x_k) \leq -c_3 \| \delta x_k \|_2^2 \label{eq:proof:lyapunov:decreasing} \\
    \left\| \frac{\partial V(A_\delta \delta x_k + \varepsilon(\delta x_k))}{\partial \varepsilon} \right\|_2 \leq c_4 \| \delta x_k \|_2. \label{eq:proof:lyapunov:lipschitz}
    \end{gather}
\end{subequations}
The goal is to show that $V(\delta x_k)$ is also a Lyapunov function for the linear system \eqref{eq:proof:linear}.
To this end, let us add and subtract $V(A_\delta \delta x_k)$ from the left-hand side of \eqref{eq:proof:lyapunov:decreasing}, leading to
\begin{equation}\label{eq:proof:lyapunov:int1}
\begin{aligned}
    & V(A_\delta \delta x_{k}) - V(\delta x_k)  + \big[ V(A_\delta \delta x_{k} + \varepsilon(\delta x_k)) - V(A_\delta \delta x_{k}) \big] \\
    & \quad \leq -c_3 \| \delta x_k \|_2^2
\end{aligned}
\end{equation}
In light of \eqref{eq:proof:lyapunov:lipschitz} and \eqref{eq:proof:linearization_err_bound}, $V(A_\delta \delta x_{k} + \varepsilon(\delta x_k)) - V(A_\delta \delta x_{k})$ can be bounded as
\begin{equation} \label{eq:proof:lyapunov:int2}
\begin{aligned}
    & \big\| V(A_\delta \delta x_{k} + \varepsilon(\delta x_k)) - V(A_\delta \delta x_{k}) \big\|_2  \\
    & \qquad \leq \left\| \frac{\partial V(A_\delta \delta x_k + \varepsilon(\delta x_k))}{\partial \varepsilon} \right\|_2 \, \| \delta x_k \|_2 \\
    & \qquad \leq c_4 L_\varepsilon \| \delta x_k \|_2^3.
\end{aligned}
\end{equation}

Owing to the bound \eqref{eq:proof:lyapunov:int2} and to the exponential GAS \eqref{eq:proof:exponential_gas}, recalling that $\lambda \in (0, 1)$, from \eqref{eq:proof:lyapunov:int1} it holds that
\begin{equation}
\begin{aligned}
    V(A_\delta \delta x_{k}) - V(\delta x_k) &\leq - c_3 \| \delta x_k \|_2^2 - c_4 L_\varepsilon \| \delta x_k \|^3_2  \\
    &\leq - c_3 \| \delta x_k \|_2^2 - \rho c_4 L_\varepsilon \| \delta x_0 \|_2 \| \delta x_k \|_2^2  \\
    &\leq - \big(c_3 - \rho c_4 L_\varepsilon \| \delta x_0 \|_2 \big) \| \delta x_k \|_2^2.
\end{aligned}
\end{equation}
Hence, there exist constants $c_5>0$ and $r_0 > 0$ such that, $\forall \delta x_0 \in \{ \delta x_0 : \| \delta x_0 \|_2 \leq r_0 \}$,
\begin{equation*}
    V(A_\delta \delta x_{k}) - V(\delta x_k) \leq - c_5 \| \delta x_k \|_2^2.
\end{equation*}
The asymptotic stability of the linear system \eqref{eq:proof:linear} is proven by using $V(\delta x_k)$ as Lyapunov function, which implies that $A_\delta$ is Schur stable. 

\end{proof}

\end{document}